\documentclass[11pt]{llncs}
\usepackage{amsmath,amssymb,amsfonts}
\usepackage{gastex}
\usepackage{graphicx}
\usepackage{epic}

\title{On two Algorithmic Problems \\about Synchronizing Automata}
\author{Mikhail V. Berlinkov\thanks{Supported by the Presidential Programme for young researchers, grant
MK-3160.2014.1 and by the Russian Foundation for Basic research,
grant 13-01-00852}}

\authorrunning{Mikhail V. Berlinkov}

\tocauthor{Mikhail V. Berlinkov (Ekaterinburg, Russia)}

\institute{ Institute of Mathematics and Computer Science,\\
              Ural Federal University, 620000 Ekaterinburg, Russia\\
              \email{berlm@mail.ru}}

\DeclareSymbolFont{rsfscript}{OMS}{rsfs}{m}{n}
\DeclareSymbolFontAlphabet{\mathrsfs}{rsfscript}

\newcommand{\sw}{reset word}

\newcommand{\sa}{synchronizing automata}
\newcommand{\san}{synchronizing automaton}
\newcommand{\rl}{reset threshold}
\newcommand{\rt}{reset threshold}

\begin{document}
\maketitle

\begin{abstract}
Under the assumption $\mathcal{P} \neq \mathcal{NP}$, we prove that
two natural problems from the theory of synchronizing automata
cannot be solved in polynomial time. The first problem is to decide
whether a given reachable partial automaton is synchronizing. The
second one is, given an $n$-state binary complete synchronizing
automaton, to compute its reset threshold within performance ratio
less than $d \ln{(n)}$ for a specific constant $d>0$.
\end{abstract}

\section{Testing for synchronization}
\label{sec_testing}


A \emph{deterministic finite automata} (DFA) $\mathrsfs{A}$ is a
triple $\langle Q,\Sigma,\delta \rangle$ where $Q$ is the state set,
$\Sigma$ is the input alphabet and $\delta: Q \times \Sigma
\rightarrow Q$ is the \emph{transition function}. If $\delta$ is
totally defined on $Q \times \Sigma$ then $\mathrsfs{A}$ is called
\emph{complete}, otherwise $\mathrsfs{A}$ is called \emph{partial}.

The transition function can be naturally extended to $\Sigma^{*}$ as
follows. For every state $q \in Q$ and $u \in \Sigma^{*}$ we let
$\delta(q,\lambda) = q$ where $\lambda$ is an empty word, and we
inductively define $\delta(q,ua) = \delta(\delta(q,u),a)$ for any $a
\in \Sigma$ provided that $\delta(q,u)$ and $\delta(\delta(q,u),a)$
are defined. We can simplify the notation by writing $S.w$ instead
of $\{\delta(q,w) \mid q \in S \}$ for a subset $S\subseteq Q$ and a
word $w \in \Sigma^*$.

A DFA $\mathrsfs{A} = \langle Q,\Sigma,\delta \rangle$ is called
\emph{synchronizing} if there exists a word $w\in\Sigma^*$ such that
$|Q.w|=1$. Notice that here, in contrast to some other versions of
synchronizability studied in the realm of partial automata (see
e.g.~\cite{Mart10}), $w$ is not assumed to be defined at all states.
Each word $w$ with this property is said to be a \emph{reset} word
for $\mathrsfs{A}$. The minimum length of reset word is called the
\emph{reset threshold} of $\mathrsfs{A}$ and denoted by
$rt(\mathrsfs{A})$. Analogously, a word $w$ \emph{synchronizes} a
subset $S \subseteq Q$ if $|S.w|=1$.

Recall that a DFA $\mathrsfs{A} = \langle Q,\Sigma,\delta \rangle$
is called \emph{reachable} if one can choose an \emph{initial} state
$q_0 \in Q$ and a \emph{final} set of states $F \subseteq Q$ such
that each state $q \in Q$ is accessible from $q_0$ and co-accessible
from $F$, i.e. there are words $u,v \in \Sigma^*$ such that $q_0.u =
q$ and $q.v \in F$. This case is of certain interest due to its
applications in dna-computing, namely, reset words for partial
reachable automata serve as \emph{constants} for the corresponding
\emph{splicing} systems (see e.g.~\cite{SplSyst}). It is known that
the problem of testing whether or not a given strongly connected
partial automaton is synchronizing can be solved in polynomial time
(see ~\cite[Algorithm~3]{Emach}). In contrast, we show here that the
problem becomes $PSPACE$-complete if we allow automata to be
reachable instead of being strongly-connected.

Now we adapt the following results from~\cite{SubSynch} about subset
synchronization in complete strongly connected automata to our case.
\begin{theorem}[\text{\cite[Theorem~7]{SubSynch}}]
\label{th_voj} There is a series of strongly connected binary
automata $\mathrsfs{A}_n=\langle Q_n,\{a,b\},\delta_n \rangle$ and
corresponding subsets $S_n \subset Q_n$ such that the minimum length
of synchronizing words for $S_n$ in $\mathrsfs{A}_n$ has order
$2^{\Omega(n)}$.
\end{theorem}

\begin{theorem}[\text{\cite[Theorem~10]{SubSynch}}]
\label{th_voj2}Given a strongly connected binary automaton
$\mathrsfs{A}=\langle Q,\{a,b\},\delta \rangle$ and a subset of
states $S \subset Q$, it is $PSPACE$-complete to decide whether or
not $S$ can be synchronized in $\mathrsfs{A}$.
\end{theorem}

Let us present a transparent reduction to the problem of
synchronization of partial automata in the following lemma.
\begin{lemma}
\label{red_lemma}For each complete strongly connected binary
automaton $\mathrsfs{A}=\langle Q, \{a,b\},\delta \rangle$ and a
subset $S \subseteq Q$ one can construct in $O(|Q|)$ time a
reachable partial $3$-letter automaton $\mathrsfs{B}=\langle
Q',\{a,b,c\},\delta' \rangle$ with at most $2|Q|$ states such that:
\begin{enumerate}
    \item If $u \in \{a,b\}^{*}$ synchronizes $S$ in $\mathrsfs{A}$ then the word $c u$ synchronizes $\mathrsfs{B}$;
    \item If $w \in \{a,b,c\}^{*}$ synchronizes $\mathrsfs{B}$ then $w$ has a suffix $u \in \{a,b\}^{*}$
such that $u$ synchronizes $S$ in $\mathrsfs{A}$;
    \item $S$ can be synchronized in $\mathrsfs{A}$ if and only if
    $\mathrsfs{B}$ is synchronizing;
    \item If $S$ can not be synchronized in $\mathrsfs{A}$ by words
    of length less than $R$ then the reset threshold of
    $\mathrsfs{B}$ is at least $R$.
\end{enumerate}
\end{lemma}
\begin{proof}
Denote $k=|S|$ and let $S = \{s_0, s_1, \dots, s_{k-1}\}$ and $Q' =
Q \cup Z$ where $Z = \{z_0, z_1, \dots, z_{k-1} \}$ is the set of
new states.

Define the transition function $\delta'$ as follows.
\begin{equation}\delta'(q,x) = \begin{cases}
\delta(q,x),\ & q \in Q, x \in \{a,b\}; \\
s_{i},\ & q = z_i, x = c; \\
z_{(i+1 \mod k)}, \ & q = z_{i}, x \in \{a,b\}.
\end{cases}
\end{equation}
Since $\mathrsfs{A}$ is strongly connected and there is the cycle by
$Z$, $\mathrsfs{B}$ is reachable for $q_0 = z_0$ and $F = Q$.

Since $Q'.c = S \subseteq Q$ while $c$ is undefined on $Q$ and the
action of the letters $a,b$ coincides in $\mathrsfs{A}$ and
$\mathrsfs{B}$ on $Q$, Claim~1 follows. Since $w$ synchronizes
$\mathrsfs{B}$ while $Z.a = Z.b = Z$ and $c$ is undefined on the
states from $Q$, we conclude that $w$ should be of the form $v c u$
where $u,v \in \{a,b\}^*$. Since $Q'.vc = S$, the word $u$ should
synchronize $S$ and Claim~2 follows. Claims~3 and 4 immediately
follow from Claims~1 and 2.
\end{proof}


\begin{theorem}
\label{thm_4} Testing a given reachable partial $3$-letter automaton
for synchronization is $PSPACE$-complete. There is a series of
reachable partial $3$-letter $n$-state automata with reset
thresholds of order $2^{\Omega(n)}$.
\end{theorem}
\begin{proof}
The problem is $PSPACE$-hard by Theorem~\ref{th_voj2} and Claim~3 of
Lemma~\ref{red_lemma}. The corresponding series of reachable
automata exists by Theorem~\ref{th_voj} and Claim~4 of
Lemma~\ref{red_lemma}.

It remains to prove that the problem belongs to $PSPACE$. Given a
reachable partial automaton $\mathrsfs{A}=\langle Q,\Sigma,\delta
\rangle$ we store only a current subset $S$ initialized by $Q$. In
an endless loop, we nondeterministically choose a letter $a \in
\Sigma$ and let $S := S.a$. If at some step $|S| = 1$ we return
``yes'', otherwise continue the iteration. Since for this algorithm
$O(n)$ memory is enough, we have an $NPSPACE$ algorithm which is
$PSPACE$ by the Savitch's theorem~\cite{Savitch}.
\end{proof}

The following lemma relies on the usual technique of encoding
letters in states (see e.g.~\cite{MyTOCS2013}).
\begin{lemma}
\label{red_2lemma}For each reachable $d$-letter partial automaton
$\mathrsfs{A}= \langle Q,\Sigma,\delta \rangle$, one can construct
in polynomial time a binary reachable partial automaton
$\mathrsfs{B}$ such that $\mathrsfs{A}$ is synchronizing if and only
if $\mathrsfs{B}$ is synchronizing, $|Q'| = d|Q|$ and
$rt(\mathrsfs{A}) \leq rt(\mathrsfs{B}) \leq d*rt(\mathrsfs{A})+1$
if $\mathrsfs{A}$ is synchronizing.
\end{lemma}
\begin{proof}
Let $\Sigma=\{a_1,a_2, \dots ,a_d\}$ we construct $\mathrsfs{B}=
\langle Q',\{a,b\}, \delta' \rangle$ as follows. We let $Q' = Q \
\times \Sigma$ and define the transition function $\delta':Q'\times
\{a,b\}\to Q'$ as follows:
\begin{equation}
\label{delta_def} \delta'((q,a_i),x) = \begin{cases}
(q,a_{\min(i+1,d)})  \ & \text{if } x=a,\\
(\delta(q,a_{i}), a_1)  \ & \text{if $x=b$ and $a_{i}$ is defined on
$q$},\\
\text{undefined otherwise}.
\end{cases}
\end{equation}

Thus, the action of $a$ on a state $q'\in Q'$ substitutes an
appropriate letter from the alphabet $\Sigma$ of $\mathrsfs{A}$ for
the second component of $q'$ while the action of $b$ imitates the
action of the second component of $q'$ on its first component and
resets the second component to $a_1$.

Given a word $w = a^{i_1} b a^{i_2} b \dots a^{i_k} b \in
\{a,b\}^{*}$, let $r(w)$ be a \emph{reduced} word
$a^{\min(i_1,d)}ba^{\min(i_2,d)}b \dots a^{\min(i_k,d)}b$ in
$\{a,b\}^*$. Besides that, we define the map $f: \{a,b\}^{*} \mapsto
\Sigma^{*}$ by $f(w) = a_{\min(i_1,d)}a_{\min(i_2,d)} \dots
a_{\min(i_k,d)}$. Given $w \in \{a,b\}^{*}$, by the definition of
$f$ we get that if $f(w)$ resets $\mathrsfs{A}$ then $bw$ resets
$\mathrsfs{B}$, and if $w$ resets $\mathrsfs{B}$ then $f^{-1}(r(w))$
resets $\mathrsfs{A}$. The lemma follows.
\end{proof}

As a straightforward corollary of Theorem~\ref{thm_4} and
Lemma~\ref{red_2lemma} (for $d=3$) we get the main result of this
section.
\begin{corollary}
\label{cor_2}Testing a given reachable partial binary automaton for
synchronization is $PSPACE$-complete. There is a series of reachable
partial binary $n$-state automata with reset thresholds of order
$2^{\Omega(n)}$.
\end{corollary}

\section{Approximation of reset thresholds}
\label{sec_find_rt}

In this section we restrict ourself to the case of complete
automata. For this case, testing for synchronization is polynomial.
When an automaton is synchronizing, the next natural problem is to
calculate its reset threshold. It is known that a precise
calculation of the \rl\ is computationally hard (see
e.g.~\cite{Ep90},\cite{OU10}). There are some polynomial time
algorithms that, given a \san, find a \sw\ for it, see,
e.g.~\cite{Ep90}. These algorithms can be used as approximation
algorithms for calculating the \rl, and it is quite natural to ask
how good such a polynomial approximation can be. The quality of an
approximation algorithm is measured by its performance ratio, which
for our problem can be defined as follows. Let $K$ be a class of
\sa. We say that an algorithm $M$ \emph{approximates the \rl\ in
$K$} if, for an arbitrary DFA $\mathrsfs{A}\in K$, the algorithm
calculates an integer $M(\mathrsfs{A})$ such that $M(\mathrsfs{A})
\ge rt(\mathrsfs{A})$. The \emph{performance ratio} of $M$ at
$\mathrsfs{A}$ is
$R_M(\mathrsfs{A})=\dfrac{M(\mathrsfs{A})}{rt(\mathrsfs{A})}$. The
author~\cite{MyTOCS2013} proved that, unless $\mathcal{P} =
\mathcal{NP}$, for no constant $r$, a polynomial time algorithm can
approximate the \rl\ in the class of all binary \sa\ with
performance ratio less than $r$.

When no polynomial time approximation within a constant factor is
possible, the next natural question is whether or not one can
approximate within a logarithmic factor. Gerbush and
Heeringa~\cite{Gerb1} conjectured that if $\mathcal{P} \ne
\mathcal{NP}$, then there exists $\alpha>0$ such that no polynomial
time algorithm approximating the \rl\ in the class of all \sa\ with
a fixed number $k>1$ of input letters achieves the performance ratio
$\alpha\log|Q|$ at all DFAs $\langle Q,\Sigma,\delta \rangle$. Using
a reduction from the problem Set-Cover and a powerful
non-approximation result from~\cite{AMS6}, Gerbush and Heeringa
proved a weaker form of this conjecture when the number of input
letters is allowed to grow with the state number.

Here we prove the conjecture from~\cite{Gerb1} in its full
generality, for each fixed size $k>1$ of the input alphabet. Though
we depart from the same reduction from Set-Cover as in~\cite{Gerb1},
we use not only the result from~\cite{AMS6}, but also some
ingredients from its proof, along with an encoding of letters in
states.

Let us follow~\cite{AMS6} and~\cite{Gerb1} below. Given a universe
$\mathcal{U} = \{u_1, \dots ,u_n\}$ and a family of its subsets,
$\mathcal{S} = \{S_1, \dots ,S_m\} \subseteq P(\mathcal{U})$ such
that $\bigcup_{S_j \in \mathcal{S}}S_j = \mathcal{U}$, Set-Cover is
the problem of finding there a minimal sub-family $\mathsf{C}
\subseteq \mathcal{S}$ that covers the whole universe in the sense
that $\bigcup_{S_j \in \mathsf{C}}S_j = \mathcal{U}$. Denote the
size of the minimal sub-family by $OPT(\mathcal{U},\mathcal{S})$.
Set-Cover is a classic $\mathcal{NP}$-hard combinatorial
optimization problem, and it is known that it can be approximated in
polynomial time to within $\ln{(n)} - \ln{(\ln{(n)})} + \Theta(1)$
(see~\cite{APP1,APP2}).

The following transparent reduction from Set-Cover is presented
in~\cite{Gerb1}. Given a Set-Cover instance
$(\mathcal{U},\mathcal{S})$, define the automaton
$$\mathrsfs{A}(\mathcal{U},\mathcal{S}) = \langle \mathcal{U} \cup
\{\hat{q}\}, \Sigma = \{a_1, \dots a_m\} \rangle$$ where the
transition function is defined as follows.
\begin{equation}\delta(u,a_i) = \begin{cases}
\hat{q},\ & u \in S_i\\
u, \ & u \notin S_i.
\end{cases}
\end{equation}

\begin{remark}
\label{rem_reduction} Let $\mathrsfs{A} = \langle Q,\Sigma,\delta
\rangle $ be the automaton defined by $(\mathcal{U},\mathcal{S})$ as
above. Then $rt(\mathrsfs{A}) = OPT(\mathcal{U},\mathcal{S}), \quad
|Q| = |\mathcal{U}|+1, \quad |\Sigma| = |\mathcal{S}|.$
\end{remark}

The following powerful result has been obtained in~\cite{AMS6}.
\begin{theorem}[\text{\cite[Theorem~7]{AMS6}}]
\label{th_AMS} Unless $\mathcal{P} = \mathcal{NP}$, no polynomial
time algorithm can approximate Set-Cover within performance ratio
less than $c_{sc} \ln{n}$ where $n$ is the size of the universe and
$c_{sc} > 0.2267$ is a specific constant.
\end{theorem}

Here we prove the aforementioned conjecture from~\cite{Gerb1} by
encoding binary representation of letters in states and using some
properties from the proof of Theorem~\ref{th_AMS}.
\begin{lemma}
\label{encoding_lem} For every $m$-letter \san\
$\mathrsfs{A}=\langle Q,\Sigma,\delta \rangle$, there is a
$2$-letter \san\ $\mathrsfs{B} = \mathrsfs{B}(\mathrsfs{A}) =
\langle Q',\{0,1\}, \delta' \rangle$ such that
\begin{equation*}
\label{limits} rt(\mathrsfs{A}) \lceil \log_2{m}+1 \rceil \le
rt(\mathrsfs{B}) \le \lceil \log_2{m}+1 \rceil(1+rt(\mathrsfs{A})),
\end{equation*}
$\mathrsfs{B}$ has at most $4 m |Q|$ states and can be constructed
in polynomial time of $m$ and $|Q|$.
\end{lemma}
\begin{proof}
Let $\Sigma=\{a_1, \dots ,a_m\}$ and for simplicity assume that $m$
is a power of $2$, i.e. $m=2^k$ (otherwise we can add at most $m-1$
letters with trivial action without impact on the bounds). Let
$\ell: \{0,1\}^{k} \mapsto \Sigma$ be a bijective function. Set $Q'
= Q \times \{0,1\}^{\leq k}$ and define the transition function
$\delta':Q'\times \{0,1\} \to Q'$ as follows. For each $q \in Q$,
each binary sequence $w \in \{0,1\}^{\leq k}$ and each bit $x \in
\{0,1\}$, we let
\begin{equation}
\label{def_deltak} \delta'((q,w),x) = \begin{cases}
(q,wx)\ & \text{if } |w| < k;\\
(q.{\ell(w)},\lambda) \ & \text{if } |w| = k, x = 1;\\
(q,w) \ & \text{if } |w| = k, x = 0.
\end{cases}
\end{equation}
Let $u = a_{j_1} a_{j_2} \dots a_{j_t}$ be a \sw\ for
$\mathrsfs{A}$. Then the word $$1^{k+1} \ell^{-1}(a_{j_1}) 1 \dots
\ell^{-1}(a_{j_t}) 1$$ is reset for $\mathrsfs{B}$ and its length
equals $(k+1)(t+1)$. The upper bound follows.

In order to prove the lower bound it is enough to consider the
shortest binary word $u$ which synchronizes the subset $(Q,\lambda)$
in $\mathrsfs{B}$. Since $u$ is chosen shortest, $u = w_1 1 w_2 1
\dots w_r 1$ where $|w_j| = k$ for each $j \in \{1, \dots r\}$.
Indeed, after applying a word $w \in \{0,1\}^k$ to the state of the
form $(q,\lambda)$ it make no sense to apply $0$ in view of the
third choice of definition~\ref{def_deltak}. Then the word
$\ell(w_1) \ell(w_2) \dots \ell(w_r)$ resets $\mathrsfs{A}$ and the
lower bound follows. \qed
\end{proof}

Now, suppose that for some constant $d>0$, there is a polynomial
time algorithm $f_2$ such that
$$rt(\mathrsfs{B}) \leq f_2(\mathrsfs{B}) \leq d\ln{(n)} rt(\mathrsfs{B})$$
for every $2$-letter $n$-state \san\ $\mathrsfs{B}$. Then
Lemma~\ref{encoding_lem} implies that for each $m \geq 2$ there is
also a polynomial time algorithm $f_m$ such that
$$rt(\mathrsfs{A}) \leq f_m(\mathrsfs{A}) \leq d\ln{(4 n m)} (1+rt(\mathrsfs{A}))$$
for every $m$-letter $n$-state \san\ $\mathrsfs{A}$. Indeed, such
algorithm first constructs $\mathrsfs{B}(\mathrsfs{A})$ with at most
$4nm$ states as in Lemma~\ref{encoding_lem}, and then runs $f_2$ on
$\mathrsfs{B}(\mathrsfs{A})$:
\begin{equation}
\label{eq_red} rt(\mathrsfs{A}) \leq f_m(\mathrsfs{A}) =
\frac{f_2(\mathrsfs{B}(\mathrsfs{A}))}{\lceil \log_2{m}+1 \rceil}
\leq d\ln{(4nm)} (rt(\mathrsfs{A})+1).
\end{equation}
Combining this with Theorem~\ref{th_AMS} and
Remark~\ref{rem_reduction} we immediately get the following
corollary.
\begin{corollary}
\label{cor_2let} Let $g(n)$ be an upper bound on the cardinality of
the set family $\mathcal{S}$ as a function of the size of the
universe $\mathcal{U}$ from the reduction to Set-Cover
from~\cite{AMS6}. Then, unless $\mathcal{P} = \mathcal{NP}$, no
polynomial time algorithm approximates \rt\ within performance ratio
$\frac{d}{\log_{n}{g(n)}+1} \ln{(n)}$ for any $d<c_{sc}$ in the
class of all $2$-letter \sa.
\end{corollary}

Thus it suffices to find a lower bound on the size of the universe
$\mathcal{U}$ and an upper bound on the size of the family of
subsets $\mathcal{S}$ in the reduction to Set-Cover presented
in~\cite{AMS6}. Namely, we need to find a polynomial upper bound for
$g(n)$.

Due to the space limit, we shall use some notation from~\cite{AMS6}
without reproducing all definitions. First, the universe
$\mathcal{U}$ is defined as $[D] \times \Phi \times B$ where $D =
\lfloor \frac{|\Phi|}{\eta|X|} \rfloor$, $\eta$ is a constant. Hence
the rough lower bound for the size of the universe $\mathcal{U}$ is
$|\Phi|$.

The size of the family of subsets $\mathcal{S}$ is equal to $D |X|
|F| + |\Phi||F|^d$ where $F$ is a field of cardinality at most
$2^{\log_2^{1-\beta}{|X|}} \leq |X|$ and $d \geq 2$ is a positive
integer which can be taken equal $3$. Hence the upper bound for
$|\mathcal{S}|$ is $\Theta(1)|\Phi||X|^d$. Thus we get that
$$\log_{|\mathcal{U}|}{|\mathcal{S}|} \leq \frac{d + \log_{|X|}{|\Phi|}}{
\log_{|X|}{|\Phi|}}.$$ Note that $|\Phi|$ is only restricted to be
some polynomial of $|X|$, i.e. it can be chosen to be $|X|^r$ for an
arbitrary large constant $r$. As a conclusion we get the following
lemma, which gives a nice property of Set-Cover itself.
\begin{lemma}
Given any $\gamma>0$, unless $\mathcal{P} = \mathcal{NP}$, no
polynomial time algorithm approximates the Set-Cover with
performance ratio $d\ln{n}$ for any $d<c_{sc}$ in the class of all
Set-Cover instances $(\mathcal{U},\mathcal{S})$ satisfying
$\log_{|\mathcal{U}|}{|\mathcal{S}|} \leq 1 + \gamma$.
\end{lemma}

Combining this with Corollary~\ref{cor_2let} gives us the second
main result.
\begin{theorem}
\label{main_res} Unless $\mathcal{P} = \mathcal{NP}$, no polynomial
time algorithm approximates the \rl\ within performance ratio less
than $0.5 c_{sc} \ln{n}$ in the class of all $n$-state \sa\ with $2$
input letters.
\end{theorem}

Let us notice that the same bound holds true for any fixed
non-singleton alphabet. Theorem~\ref{main_res} improves the previous
result of the author~\cite{MyTOCS2013} about non-approximability
within any constant factor and also gives the positive answer to the
corresponding conjecture from~\cite{Gerb1}.

It is known (see e.g.~\cite{APP1},\cite{APP2}) that the greedy
algorithm for Set-Cover has a logarithmic performance ratio. Despite
of relations with the problem of computing the reset threshold,
there is a series of automata for which the greedy algorithm =
computes reset threshold with linear performance ratio ([Ananichev,
2014], unpublished). Hence the first natural open question is about
the tightness of the bound in Theorem~\ref{main_res}.

\medskip
\noindent\textbf{Acknowledgements.} The author thanks the anonymous
referees for their useful remarks and suggestions.


\begin{thebibliography}{9}
\bibitem{AMS6}
Alon, N., Moshkovitz, D., Safra, S.: Algorithmic Construction of
Sets for k-restrictions. ACM Trans. Algorithms, 2(2), pp. 153--177
(2006)

\bibitem{MyTOCS2013}
Berlinkov, M.: Approximating the Minimum Length of Synchronizing
Words Is Hard. Theory Comput. Syst. 54(2), pp. 211--223 (2014)


\bibitem{SplSyst}
Bonizzoni, P., Jonoska, N: Regular Splicing Languages Must Have a
Constant. In: Mauri, Giancarlo and Leporati, Alberto (eds)
Developments in Language Theory, Lect.\ Notes Comp.\ Sci., 6795, pp.
82--92, Springer Berlin Heidelberg (2011)

\bibitem{Ep90}
Eppstein, D.: Reset Sequences for Monotonic Automata. SIAM J.
Comput. 19, pp. 500--510, (1990)

\bibitem{Gerb1}
Gerbush, M., Heeringa, B.: Approximating Minimum Reset Sequences.
15-th Implementation and application of automata, Lect.\ Notes
Comp.\ Sci. 6482, pp. 154--162, Springer, Berlin (2011)

\bibitem{Waut_synch}
Ivan, S: Synchronizing Weighted Automata, arXiv:1403.5729 (2014)

\bibitem{APP1}
Lov\'{a}sz, L.: On the Ratio of Optimal Integral and Fractional
Covers. Discrete Mathematics, 13: pp. 383--390 (1975)

\bibitem{Mart10}
Martugin, P: Complexity of Problems Concerning Carefully
Synchronizing Words for PFA and Directing Words for NFA, Lect.\
Notes Comp.\ Sci., 6072, pp. 288--302 (2010)


\bibitem{OU10}
Olschewski, J., Ummels, M.: The Complexity of Finding Reset Words in
Finite Automata. Lect.\ Notes Comp.\ Sci. 6281, pp. 568--579 (2010)

\bibitem{Savitch}
Savitch, W: Relationships Between Nondeterministic and Deterministic
Tape Complexities, Journal of Computer and System Sciences 4 (2),
pp. 177-–192 (1970)

\bibitem{APP2}
Slavik, P.: A Tight Analysis of the Greedy Algorithm for Set Cover.
In Proc. 28th ACM Symp. on Theory of Computing, pp. 435--441 (1996)


\bibitem{Emach}
Travers, N., Crutchfield, J.: Exact Synchronization for Finite-State
Sources, J. Stat. Phys. 145:5, pp. 1181--1201 (2011)

\bibitem{SubSynch}
Vojt\v{e}ch, V: Subset Synchronization of Transitive Automata,
arXiv:1403.3972 (accepted to AFL 2014).


\end{thebibliography}
\end{document}